\documentclass{article}

\usepackage{arxiv}

\usepackage[utf8]{inputenc} 
\usepackage[T1]{fontenc}    
\usepackage[hidelinks]{hyperref}       
\usepackage{url}            
\usepackage{booktabs}       
\usepackage{amsfonts}       
\usepackage{nicefrac}       
\usepackage{microtype}      
\usepackage{graphicx}
\usepackage{doi}

\usepackage{amsmath}
\usepackage{amssymb}
\usepackage[mathcal]{euscript}

\usepackage{bm}
\usepackage{epsfig}
\usepackage{color}
\usepackage{cite}
\usepackage{amsthm}

\usepackage{enumitem}

\usepackage{cleveref}

\newtheorem{lemma}{lemma}

\usepackage{chngcntr}
\usepackage{apptools}

\usepackage[mathcal]{euscript}

\AtAppendix{\counterwithin{lemma}{section}}

\makeatletter \@addtoreset{lemma}{section} \makeatother

\makeatletter \@addtoreset{table}{section} \makeatother

\title{A Novel First-Principles Model of Injection-Locked  Oscillator 
Phase Noise.}

\author{Torsten Djurhuus \thanks{The authors are
with the Institute of Physics, Goethe University of Frankfurt am
Main, Max-von- Laue-Strasse 1, 60438, Frankfurt am Main.
(correspondence e-mail:
t.djurhuus@physik.uni-frankfurt.de).} \\
  Goethe-University Frankfurt\\
  \texttt{t.djurhuus@physik.uni-frankfurt.de} \\
  \And
  Viktor Krozer \\
  Goethe-University Frankfurt\\
  \texttt{krozer@physik.uni-frankfurt.de}\\
}

\renewcommand{\headeright}{\textit{arXiv} Preprint}
\renewcommand{\undertitle}{\textit{arXiv} Preprint}

\hypersetup{ pdftitle={A Novel First-Principles Model of Injection-Locked  Oscillator 
Phase Noise.}, pdfsubject={eess.SY, math.CA,
math.DS}, pdfauthor={Torsten Djurhuus, Viktor Krozer},
pdfkeywords={injection-locked oscillators, phase noise,
circuit analysis, nonlinear dynamical systems, system analysis and
design}, }

\begin{document}

\maketitle

\begin{abstract}
  The paper documents the development of a novel time-domain model 
  of injection-locked oscillator phase-noise response. The 
  methodology follows a first-principle approach 
  and applies to all circuit topologies, coupling configurations, 
  parameter dependencies \emph{etc.} 
  The corresponding numerical algorithm is readily integrated into 
  all major commercial simulation software suites. 
  The model advances current state-of-the-art pertaining to 
  analytical modelling of this class of circuits.  Using this 
  novel analytical framework, several 
  important new insights are revealed which, in-turn, translate into 
  useful design rules for synthesis of injection-locked oscillator 
  circuits with optimal noise performance.    
\end{abstract}


\section{Introduction}
\label{sec0}

Efficient and rigorous modelling tools for predicting 
phase-noise (PNOISE) response of oscillator/clock-circuits 
play a critical role in the design cycle of 
modern communication and remote-sensing systems. 
Given the complexity of contemporary device models, 
including the presence of various correlated, colored and modulated noise 
sources, and considering the scale of modern circuit schematics, 
such design work necessitates the use of an 
Electronic-Design-Automation (EDA) simulation 
environment (\emph{e.g.} 
Keysight-ADS\textsuperscript{\textcopyright} or
Cadence SpectreRF\textsuperscript{\textcopyright}). 
Relevant numerical algorithms must hence be compatible with this 
EDA interface. Simply stated, this requires a methodology 
which is formulated directly from a set of unspecified, 
nonlinear stochastic-differential-equations (SDE). 
The resulting model hence applies to any 
circuit regardless of topology, system dimension, parameter 
dependencies \emph{etc.} Herein, this modelling strategy
is referred to as a \emph{first-principle} (FP) approach, although sometimes 
we may use terms such as \emph{e.g.} unified, generalized or macro-model to 
describe the same concept. 
\par
The paper presents a novel time-domain (TD) FP methodology aimed at 
describing the PNOISE response of injection-locked oscillator (ILO) circuits; the so-called 
ILO phase macro-model (ILO-PMM). The  model is derived as a specialization 
of an earlier model published by the authors \cite{djurhuus22}.
The methodology is highly rigorous, being based on 
nonlinear stochastic integration techniques and Floquet 
decomposition methods. It represents a direct extension 
of the single oscillator phase-macro-model (PMM) developed in \cite{kartner1990,demir2000,
traversa2011}. The ILO-PMM, however, not only extends but also replaces the standard 
single oscillator PMM; reverting to this representation for zero coupling. The ideas 
presented herein advance the current state-of-the-art (SOA) \emph{w.r.t} modelling, 
analysis, synthesis and optimization of ILO noise response. The derivation of 
the ILO-PMM framework will be discussed in \cref{sec1b} following a brief introduction 
to the underlying theory in \cref{sec1a}.  
\par
Existing FP solutions to the problem discussed here are generally 
formulated as frequency-domain (FD) models involving some variation of the standard 
conversion-matrix (C-MATRIX) methodology \cite{samori1998}. 
These types of methods, such as \emph{e.g.} the 
\emph{pnmx} algorithm which is part of the Keysight-ADS\textsuperscript{\textcopyright} 
software suite, are entirely numerical in nature with no accompanying
representation for analysis purposes; herein referred to as \emph{black-box models}. In contrast, ILO-PMM 
methodology produces a unique closed-form algebraic expression representing the 
ILO PNOISE response. This analytical interface represents an unprecedented and unparalleled 
feature of the ILO-PMM as no other FP description (EDA compatible model), 
currently published, exist with this property. 
\par
The ILO-PMM, being an analytical FP framework both 
extends and replaces all previous phenomenological  
analytical methodologies published on this topic. Such schemes are 
all predicated on the assumption of special ILO solution (\emph{e.g} 
the quasi-sinusoidal Kurokawa methodology \cite{kurokawa1968}) or a special 
circuit topology  (\emph{e.g.} a block-diagram formulation). 
These reduced-order, empirical 
modelling frameworks cannot readily encompass the circuit complexities mentioned above and 
are obviously not compatible with the EDA interface. Nevertheless, they 
do, however, represent the current state-of-the-art (SOA) when it comes to 
mathematical/theoretical analysis of the ILO PNOISE scenario; simply because, 
up until this point, no fully rigorous alternatives existed. By preempting these 
phenomenological methodologies the ILO-PMM significantly advances 
the current SOA in the realm of circuit theoretical analysis. The text in  \cref{sec2,sec2a,sec2b} discuss how 
the ILO-PMM can be reduced/downsized to create a equivalent FP representation 
of the well-established Kurokawa quasi-sinusoidal (Q-SINUS) model representation 
\cite{kurokawa1968,ramirez2008}; referred to as the K-ILO model. 
Comparing the ILO-PMM and the reduced-order K-ILO model descriptions, several 
new insights, pertaining to the \emph{bounds} or \emph{range-of-application} 
of the Kurokawa methodology, are uncovered and discussed 
(see \cref{sec2b:tab1,sec2b:lem1}). 
\par
In \cref{sec3} below the capability of the proposed ILO-PMM framework 
is demonstrated by comparing its performance against well-established 
commercial numerical simulation routines. The simulations are carried out  
on a $0.9\mathrm{GHz}$ circuit created by coupling a simple LC, negative-resistance 
oscillator to a CMOS, cross-coupled LC-tank unit through a unilateral buffer 
amplifier.  The analytical ideas developed in 
\cref{sec2b} are tested by comparing the outputs of the 
ILO-PMM and reduced-order K-ILO models, applied to the circuit described above.  
The work in \cref{sec3} verifies the ILO-PMM model, and the underlying 
theoretic ideas used to develop it. The woek in both \cref{sec2,sec3} 
solidifies the ILO-PMM status as the new benchmark \emph{w.r.t.} 
theoretical analysis of the ILO noise scenario. The novel analytical 
tools introduced in this paper will prompt new insights and ideas 
leading to the development of design rules 
aimed at synthesizing ILO circuits with optimal noise performance.

\section{Theory}

\label{sec1}

The topic of discussion in this paper is the ILO-PMM, 
which refers to a TD, first-principles (FP) model of ILO PNOISE response. 
This novel result relies on theory developed by the authors 
in an earlier publication \cite{djurhuus22}. That paper treated the 
scenario of a general coupled oscillator ensemble. 
Below, a brief introduction to the rather complicated 
theoretic program, discussed in \cite{djurhuus22}, is given. 
We then proceed show how to utilize these ideas 
to formulate the novel ILO-PMM model.

\subsection{ The Coupled Oscillator PMM : a Brief Review.}

\label{sec1a}

\begin{figure}[!h]
\begin{center}
\includegraphics[scale=0.8]{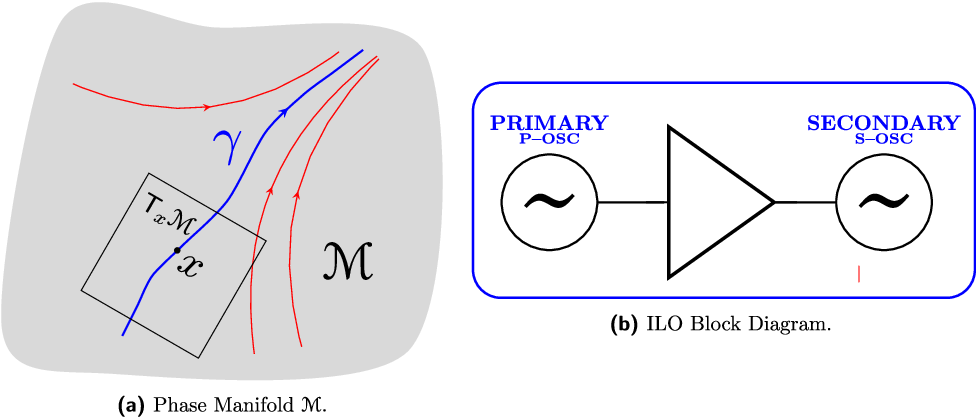}
\end{center}
\caption{ (a) : The limit-cycle $\gamma$ (blue orbit) is embedded in
the $k$-dimensional, closed, \emph{phase-manifold}, $\mathcal{M}$. 
All orbits on, and off, $\mathcal{M}$ approach this
$1$-dimensional set asymptotically with time (red orbits). The
tangent-space to $\mathcal{M}$ at a point $x\in \mathcal{M}$,
$\mathsf{T}_x\mathcal{M}$, is an affine copy of $\mathbb{R}^k$. 
(b) : Illustrating the ILO circuit configuration. The 
free-running primary oscillator (P-OSC) is coupled unilaterally, through 
some buffer amplifier to the secondary oscillator (S-OSC). Assuming a synchronized PSS is reached, 
the S-OSC is locked to P-OSC injected signal.
} \label{sec1a:fig1}
\end{figure}

A novel FP (unified) model, for the prediction of the first-order stochastic response of 
noise-perturbed coupled oscillator ensemble, was recently 
published in \cite{djurhuus22} by the authors. One of the results which was achieved  
involved a unique closed-form expression for the PNOISE spectrum
of a general ensemble. The model developed 
in \cite{djurhuus22} will be referred to as the coupled oscillator 
macro-model (COSC-PMM). The COSC-PMM framework is formulated from FP, 
and thus, by definition, incorporates all circuit topologies,
coupling configurations and parameter dependencies in one closed-form 
expression. The model extends and replaces the well-established single oscillator 
PMM model \cite{demir2000,kartner1990,traversa2011}. The novel COSC-PMM theoretic 
framework was achieved by employing ideas and results 
from various branches of mathematics 
such as \emph{e.g.}  manifold-theory, differential-geometry and 
Floquet theory. Below, a brief review of these rather involved 
topics is provided. For a more detailed discussion of these issues 
the reader is referred to the text in \cite{djurhuus22}.
\par
Consider $k$ free-running and asymptotically stable autonomous oscillator units
coupled through some type of network; referred to herein 
as a $k$-ensemble. We assume that this network reaches a synchronized state. 
Let $x\in \mathbb{R}^n$ be the $n$-dimensional
state-space of the coupled network ($k$-ensemble). The synchronized
PSS is then written, $x_s(t+T_0) = x_s(t)$, with $T_0>0$ being the
period of the synchronized k-ensemble. Below the $\nu$th harmonic 
of the PSS spectrum is written $X_{s,\nu}$. Here $X_{s,\nu}$ 
is an $n$-dimensional vector \emph{i.e.} $X_{s,\nu} \in \mathbb{R}^n$. 
Furthermore, the $q$th, $q\in [1;n]$, element of such a vector is 
written $X^{[q]}_{s,\nu} \in \mathbb{R}$.  
\par
In \cite{djurhuus22} the concept of a \emph{phase-manifold}, 
$\mathcal{M}$, is introduced. This is a closed $k$-dimensional 
space which is known to exist, for weak coupling, 
due to the persistence of normally hyperbolic manifolds. 
Furthermore, it can be shown that the ensemble
limit-cycle, $\gamma$, is embedded in this manifold; \emph{i.e.}
it is a sub-manifold. Let $\mathbb{T}_{\gamma}\mathbb{M}$ refer 
to the tangent-bundle, on $\mathbb{M}$, and along $\gamma$ where the 
concept of a tangent-bundle simply refers to the disjoint union 
of tangent-spaces (affine copies of vector-space $\mathbb{R}^k$) along $\gamma$; 
\emph{i.e.}
$\mathbb{T}_{\gamma}\mathbb{M} = \{\mathbb{T}_{x}\mathbb{M}\}_{x\in \gamma}$. As detailed 
in \cite{djurhuus22}, there exists exactly $k$ Floquet vectors $\{u_i(t)\}_{i=1}^k$ 
which span the bundle $\mathbb{T}_{\gamma}\mathbb{M}$. These 
$k$ modes, also known as the \emph{phase-modes}, 
govern the PNOISE response of the $k$-ensemble. The scenario discussed 
here is illustrated in \cref{sec1a:fig1}.(a).
\par
The setup described above and in \cref{sec1a:fig1}.(a), assures the 
existence of $k$ unique Floquet phase-modes $\{u_i(t)\}_{i=1}^k$. 
It is then possible to calculate the PNOISE response which 
is the response spanned by these operators. 
Following some rather lengthy and involved calculations, 
an expression for the PNOISE spectral density of a $k$-ensemble  
was derived in \cite[sec. 6]{djurhuus22} which we repeat here

\begin{equation}
\mathfrak{L}^{(\nu)}(\omega_m^{(\nu)}) = \frac{ (1 - a^{(\nu)} )(\omega_0^2c)
-2b^{(\nu)} \omega_m }{ (0.5\omega_0^2c)^2 + (\omega_m^{(\nu)})^2} +
\sum_{\rho} \sum_{l=2}^k \frac{ \Upsilon^{(\nu)}_{l\rho} \bigl[
2|\mu_{l,r}| + (\omega_0^2\rho^2c) \bigr] +
2\Delta^{(\nu)}_{l\rho}\bigl[ \omega_m^{(\nu)} + \mu_{l,i} \bigr] }{ \bigl(
|\mu_{l,r}| + \bigl(0.5\omega_0^2\rho^2c\bigr) \bigr)^2 + \bigl(
\omega_m^{(\nu)} + \mu_{l,i} \bigr)^2}   \label{sec1a:eq1}
\end{equation}

with $\omega_0 = 2\pi/T_0$ being the operating frequency of the synchronized
ensemble, superscript, index $\nu$ refers carrier harmonic around which the PNOISE 
response is calculated, 
$\mu_s = \mu_{s,r} + j \mu_{s,i} \in \mathbb{C}$, is the $s$th
characteristic Floquet exponent (real/imaginary parts) and
$\omega^{(\nu)}_m = \omega - \nu\omega_0$ is the $\nu$th harmonic
offset frequency. The real parameter, $c \in \mathbb{R}_{+}$, is known as the 
\emph{phase-diffusion} constant\cite{kartner1990,demir2000,traversa2011,djurhuus2009}. 
All sums in \cref{sec1a:eq1} w/o bounds operate in the interval $(-\infty,\infty)$. 
The operators $a^{(\nu)},b^{(\nu)},\Upsilon^{(\nu)}_{l\rho},\Delta^{(\nu)}_{l\rho}$ introduced 
in \cref{sec1a:eq1} are defined through 

\begin{align}
a^{(\nu)}+jb^{(\nu)} &= \bigl[\Omega^{(\nu)}\bigr]_{q,q}/\Vert X^{[q]}_{s,\nu}\Vert^2
\label{sec1:eq2} \\
\Upsilon^{(\nu)}_{l\rho} + j \Delta^{(\nu)}_{l\rho} &= 
\bigl[\Theta^{(\nu)}_{l\rho}\bigr]_{q,q}/
\Vert X^{[q]}_{s,\nu}\Vert^2 \label{sec1:eq3}
\end{align}  

where index $q$ refers to the index of the \emph{observation-node} \emph{i.e.} 
the node where PNOISE is measured,  $[S]_{x,y}$ denotes to the 
element at row $x$ and column $y$ of matrix $S$ and $X^{[q]}_{s,\nu}$ refers 
to the $q$ element of the PSS $\nu$th harmonic vector (see discussion above).  
The two complex tensor operators $\Omega$ and $\Theta_{l\rho} \in \mathbb{C}^{n\times n}$ 
have the form (see \cite[appendix B]{djurhuus22} for details)

\begin{align}
\Omega^{(\nu)} &=  \sum_{m=2}^k \sum_{p} \frac{
  U_{1,\nu}\Lambda_{1,0}^{\top}\Lambda_{m,\nu-p}^*U_{m,p}^{\dagger}}{j\omega_0(p-\nu)
  - \mu_m^*} \label{sec1a:eq4}\\
  \Theta_{l\rho}^{(\nu)} &= \frac{
  U_{1,\rho}\Lambda_{1,0}^{\top}\Lambda_{l,\rho-\nu}^*U_{l,\nu}^{\dagger}}{j\omega_0(\nu-\rho)
    - \mu_l^* - \mu_1} +  \frac{\sum_{i = 2}^k  \sum_p
    U_{i,p}\Lambda_{l,\rho-p}^{\top}\Lambda_{l,\rho-\nu}^*U_{l,\nu}^{\dagger}}{j\omega_0(\nu-p)
    - \mu_m^* - \mu_i}\label{sec1a:eq5}
\end{align}

with $U_{i,j} \in \mathbb{C}^n$ being the $j$th harmonic of the $i$th 
Floquet phase-mode vector $u_i(t) : \mathbb{R} \to \mathbb{C}^n$ 
and $\Lambda_{i,j} \in \mathbb{C}^n$ is the $j$th harmonic of the $i$th 
Floquet phase-mode lambda-vector $\lambda_i(t) : \mathbb{R} \to \mathbb{C}^n$ 
where $\lambda_i(t) = v_i^{\top}(t)B(x_s(t))$, with $v_i(t) : \mathbb{R} \to \mathbb{C}^n$ 
being the $i$th Floquet phase-mode dual-vector. Finally, $B(x_s(t)) : \mathbb{R} \to \mathbb{R}^{n\times p}$ is the noise modulation 
matrix, see \emph{e.g.} \cite{kartner1990,demir2000,traversa2011,djurhuus22}, 
which describes how $p$ white-noise sources are injected into the $n$-dimensional 
state-space and the manner in which the sources are modulated by the PSS, $x_s(t)$.

\subsection{An Injection Locked Oscillator PMM : the ILO-PMM.}
\label{sec1b}

The ILO configuration is illustrated in \cref{sec1a:fig1}.(b). 
As can be seen, the coupled ensemble involves 2 oscillator units, 
with the secondary oscillator (S-OSC) coupled to the 
primary oscillator (P-OSC) through 
some unilateral buffer/amplifier circuit. 
Since we are considering 2 coupled units it follows from 
the discussion in \cref{sec1a}, that there will exist two real 
Floquet phase-modes ($\mu_1=0,\mu_2<0$) governing the synchronized ILO PNOISE response. 
Considering the PNOISE spectrum around the 1st PSS harmonic spectrum ($\nu=1$), 
and given the setup with 2 phase-modes ($k=2$), the general expression 
in \cref{sec1a:eq1} specializes to the following model for an ILO circuit

\begin{equation}
\mathfrak{L}_{\text{\tiny ILO }}(\omega_m) = \frac{ (1- \alpha )(\omega_0^2c)
-2\beta \omega_m }{ (0.5\omega_0^2c)^2 + \omega_m^2} +
\sum_{\rho} \frac{ \Delta_{\rho} \bigl[
2|\mu_2| + (\omega_0^2\rho^2c) \bigr] +
2\Gamma_{\rho}\omega_m }{ \bigl(
|\mu_2| + \bigl(0.5\omega_0^2\rho^2c\bigr) \bigr)^2 +
\omega_m^2} \label{sec1b:eq1}
\end{equation}

where the notation, $\mu_2 = \mu_{2,r}$, was used since we know that 
the second phase mode is real and 

\begin{align}
\alpha + j \beta  &= \bigl[\Psi \bigr]_{q,q}/\Vert X^{[q]}_{s,1}\Vert^2  \label{sec1b:eq2}\\ 
\Delta_{\rho} + j \Gamma_{l\rho} &= \bigl[\Phi_{\rho}\bigr]_{q,q}/
\Vert X^{[q]}_{s,1}\Vert^2 \label{sec1b:eq3}
\end{align}

Here $q$, again, denotes the observation node (see discussion in \cref{sec1a}) 
and the various complex tensor operators are given through  

\begin{align}
\Psi &=  \sum_p \frac{
U_{1,1}\Lambda_{1,0}^{\top}\Lambda_{2,1-p}^*U_{2,p}^{\dagger}}{j\omega_0(p-1)
- \mu_{2}} \label{sec1b:eq4} \\
\Phi_{\rho} &= \frac{
U_{1,\rho}\Lambda_{1,0}^{\top}\Lambda_{2,\rho-1}^*U_{2,1}^{\dagger}}{j\omega_0(1-\rho)
- \mu_{2}} +  \frac{\sum_p
U_{2,p}\Lambda_{2,\rho-p}^{\top}\Lambda_{2,\rho-1}^*U_{2,1}^{\dagger}}{j\omega_0(1-p)
- 2\mu_{2}} \label{sec1b:eq5}
\end{align}

The result in \crefrange{sec1b:eq1}{sec1b:eq5} describes the ILO-PMM. 
This model represents the first ever, at-least to the author's knowledge, rigorous 
TD FP (non-phenomenological/empirical) representation of ILO PNOISE response. 
Established FP methodologies, such as \emph{e.g.} the C-MATRIX class of FD 
algorithms used in \emph{e.g.} the \emph{pnmx} routine implemented in the 
Keysight-ADS\textsuperscript{\textcopyright} EDA suite, are exclusively 
numerical in nature \emph{i.e.} \emph{black-box routines} 
(see discussion in introduction). The ILO-PMM advances the current 
SOA by providing an explicit closed-form expression for the ILO PNOISE spectrum. 
The rigorous ILO-PMM modelling framework both advances and 
repairs the flawed linearized FD FP schemes (\emph{e.g.} the  C-MATRIX methodology), 
involving, among other things, artificial singularities in the PNOISE spectrum 
\cite{demir2000}, by calculating the response through direct stochastic integration of the 
raw nonlinear circuit SDE's\cite{djurhuus22}. 
The addition of the new analytical tool-set, discussed above, will lead to the invention 
of new  design rules for synthesis of circuits with optimal noise performance.

\section{ Comparing the ILO-PMM to Established Analytical Models. }
\label{sec2}

The ILO-PMM framework, \crefrange{sec1b:eq1}{sec1b:eq5}, represents the only 
example of a FP (\emph{i.e.} non-phenomenological/empirical) analytical 
model of ILO noise response. Given this fact, 
the following question is relevant : 
why does the ILO-PMM not resemble the earlier, 
well-established analytical model representations? 
The disconnect becomes obvious when comparing 
\cref{sec1b:eq1} the established standard form of the ILO PNOISE spectral 
characteristic which can be written \cite{kurokawa1968,ramirez2008,Chang97,Chang97_2,Lynch01}

\begin{equation}
  \widehat{\mathfrak{L}}_{\text{\tiny ILO }}(\omega_m) = 
  \frac{ \Omega_{3\text{dB}}^2\mathfrak{L}_{\text{P}}(\omega_m) + N_{\text{S}}}
  { \omega_m^2 + \Omega_{3\text{dB}}^2} \label{sec2:eq1}
  \end{equation}

where $\Omega_{3\text{dB}}$ is the pole of the ILO spectrum 
$\mathfrak{L}_{\text{P}}(\omega_m) : \mathbb{R}\to \mathbb{R}$ is the PNOISE 
spectrum of the free-running P-OSC and $N_{\text{S}}\in \mathbb{R}$ is some 
scalar representing noise-contribution from the S-OSC circuit (see \cref{sec1a:fig1}.(b)). 
The standard-form characteristic in \cref{sec2:eq1} is 
generally reached using less rigorous methods (see  discussion below), however, 
this does not overturn the fact such schemes have been 
highly successful in predicting the noise response of a large 
class of important circuits. It is therefore 
essential to be able to relate and connect the ILO-PMM, developed 
herein, to these earlier well-established results. 
Below, this issue is confronted within the context of 
the single most successful and influential 
model ever published on ILO noise response : the Kurokawa model\cite{kurokawa1968}.

\subsection{The Kurukawa Methodology.}

\label{sec2a}

In 1968 by Kaneyuki Kurokawa developed a novel methodology for 
predicting the noise-response of near-sinusoidal (high-Q) ILO circuits 
perturbed by weak noise sources. 
His approach centered around a \emph{quasi-sinusoidal} 
assumption for the full noisy ILO solution and a restriction 
to a planar representation for the oscillator units (\emph{i.e.} 2-D oscillators). 
The term \emph{quasi-sinusoidal} 
(Q-SINUS) implies that the PSS can be approximated as purely 
sinusoidal (higher harmonics are discarded) with a slow-moving 
phase/amplitude envelope\footnote{The Kurokawa ILO PSS 
is written $x_s(t) = \Re\{ A_0\exp(j\omega_0t + \phi_0)\}$,
where $A_0,\phi_0 \in \mathbb{R}$ are the steady-state amplitude and phase offset 
parameters. The full noise-perturbed solution 
is then written $x(t) = x_s(t) + \delta x(t) = 
\Re\{ A(t)\exp(j\phi(t))\}$, where $\delta x(t) : \mathbb{R} \to \mathbb{R}^2$ 
is the slow-moving envelope which 
can be written in-terms of the amplitude and phase envelopes  
$\delta_{A/\phi}(t)$ \emph{i.e.} $A(t) = A_0 + \delta_A(t)$,  
and $\phi(t) = \omega_0t + \phi_0 + \delta_{\phi}(t)$. Using this notation 
a complex envelope equations can be formulated. 
Using averaging methods, real differential equations 
for the envelopes, $\delta_{A/\phi}(t)$, are derived. Finally, 
Fourier transforming these expressions the PNOISE spectrum can be calculated 
which takes the form shown in \cref{sec2:eq1}.\label{sec2a:foot1}} induced by the weak noise source drive. The novel modelling strategy, developed 
in the seminal paper \cite{kurokawa1968} (see \cref{sec2a:foot1} for a brief review), has since proved highly successful 
in developing useful models for all 
kinds of oscillator configurations (coupled and singles) driven by weak noise\cite{ramirez2008,
Chang97,Chang97_2,Lynch01,djurhuus2005,
djurhuus2005_2,djurhuus2006}. The Kurokawa methodology represents an 
example of the phenomenological/empirical modelling 
approach discussed above. It is given this label  
since it is predicated on a specific type of oscillator PSS (sinusoidal), a specific 
type of solution (\emph{i.e.} Q-SINUS), and furthermore is restricted 
to 2-D oscillator units. Below, the Q-SINUS modelling approach is briefly introduced, 
and it is then shown how to connect this important methodology to 
the novel FP representation proposed herein.

\subsubsection{Rediscovering the Kurokawa Model : Creating Reduced-Order Equivalent Model.}
\label{sec2a1}

Both the ILO-PMM framework, 
which applies to all possible ILO circuits and solutions 
(see discussion in \cref{sec1b}), 
and the Kurokawa model claim to correctly predict the PNOISE response of an 
ILO circuit generating a Q-SINUS solution. Given this fact, it   
must then be possible to somehow connect these two very different methodologies. 
Indeed, as will be shown, it is possible to construct reduced-order model 
description, the K-ILO model, from the ILO-PMM framework which constitutes an 
equivalent FP representation of the original phenomenological Kurokawa model. 
The purpose of performing this operation 
is two-fold as will now be discussed.


\begin{enumerate}
\item By producing this reduced-order K-ILO model we are effectively proving that 
the ILO-PMM supersedes and encompasses 
the Kurokawa Q-SINUS methodology. This result is not limited 
to the specific Kurokawa model but extends to 
most, if not all, non-rigorous/approximative modelling strategies. This must 
be so since the ILO-PMM represents a FP (non-phenomenological) analytical model 
which is applicable to all possible circuits and solutions. Hence, 
any such non-rigorous, phenomenological and/or empirical modelling attempt 
must be represented within the ILO-PMM framework as \underline{a reduced-order sub-model}. 
The ILO-PMM hence both \underline{extends and replaces} all these  
previously published non-rigorous analytical models. This 
result advances the SOA in the field of ILO noise analysis. It is 
of fundamental importance with far-reaching implications for future research. 
\item Developing the FP reduced-order, equivalent Kurokawa model 
gives unique insights into which components 
of the original representation are expunged or neglected in-order 
to reach this downsized version. The analytical work allows us 
to understand and quantify the \emph{bounds} of these 
phenomenological strategies. It establishes 
a \underline{range-of-application} (see \cref{sec2b:tab1,sec2b:lem1} below) 
for the Kurokawa (Q-SINUS) class of models. Outside this range, 
this type of model breaks down \emph{i.e.}
is no longer able to correctly predict the response. This is an 
important result which furthers our understanding of ILO noise modelling 
and analysis with many potential applications for practical circuit design.   
\end{enumerate}

The analysis starts with the following result.

\begin{lemma}
\label{sec2a1:lem1}

Assuming Q-SINUS operation, a reduced-order representation of the 
ILO-PMM, referred to as the K-ILO model, can be calculated. This 
model is an equivalent version of the original 
Kurokawa result\cite{kurokawa1968}. The PNOISE spectrum, 
$\mathfrak{L}_{\text{\tiny K-ILO }}(\omega_m)$, 
calculated using this reduced-order model has the form

\begin{equation}
  \mathfrak{L}_{\text{\tiny K-ILO }}(\omega_m) =
  \frac{ \Delta^{\text{\tiny (K)}}_{0} + |\mu_{2}|^2 \mathfrak{L}_{\text{P}}(\omega_m)}{ |\mu_{2}|^2 +  
  \omega_m^2}
  \label{sec2a1:eq1}
\end{equation}

with $\Delta^{\text{\tiny (K)}}_0 \in \mathbb{R}$ being a real scalar defined through

\begin{equation}
  \Delta^{\text{\tiny (K)}}_0 =
  2\bigl[\ U_{2,1}\Re\bigl\{\Lambda_{2,1}^{\top}\Lambda_{2,1}^*\bigr\}
  U_{2,1}^{\dagger} \bigr]_{q,q}/\Vert X^{[q]}_{s,1}\Vert^2 + 5\omega_0^2c
  \label{sec2a1:eq2}
\end{equation}

where all parameters and notation is explained in \cref{sec1a} (see text accompanying \cref{sec1a:eq4,sec1a:eq5}) 
and $\mathfrak{L}_{\text{P}}(\omega_m) : \mathbb{R} \to \mathbb{R}$ refers to the PNOISE 
spectrum of the free-running P-OSC circuit (see \cref{sec1a:fig1}.(b)). 

\end{lemma}
\begin{proof}
    see \cref{app1:sec1} 
\end{proof}

Comparing \cref{sec2a1:eq1}  with the standard-form expression in \cref{sec2:eq1} 
yields $|\mu_2| = \Omega_{\text{3dB}}$ and $\Delta^{\text{\tiny (K)}}_0  = 
N_s$.

\subsection{Comparing the ILO-PMM to the Reduced-Order K-ILO Model.}
\label{sec2b}

The K-ILO model, derived above in \cref{sec2a1:lem1}, is 
a \emph{reduced-order} version of the complete FP ILO-PMM scheme derived in 
\cref{sec1b}. By definition, this implies that 
the ILO-PMM must contain additional information not available in this Kurokawa 
equivalent model representation. Below we set out to quantify these ideas. 
The purpose and motivation for this work was discussed above.  
Comparing the ILO-PMM and K-ILO models, defined in \crefrange{sec1b:eq1}{sec1b:eq5} 
and \cref{sec2a1:eq1,sec2a1:eq2}, 
several interesting differences are observed. Three of these are listed 
in \cref{sec2b:tab1}  and the main conclusions from this discussion 
are summarized in \cref{sec2b:lem1} below. 

\begin{lemma}[Kurokawa Model Range-of-Application]
The Kurokawa Q-SINUS methodology, as described in \cite{kurokawa1968}, 
will fail  to correctly predict the ILO  PNOISE spectrum if 
the circuit PSS induces, \textbf{1) :} $\Lambda_{1,0} \neq 0$, \textbf{2) :} $\Lambda_{2,i} \neq  0$ for 
$i \neq \pm 1$, \textbf{3) :} $|\mu_2| \sim \omega_0^2 c$ (strong P-OSC noise drive).
\label{sec2b:lem1}   
\end{lemma}
\begin{proof}
follows directly from the discussion in \cref{sec2b:tab1} and the fact that 
the Kurokawa and K-ILO model representations are equivalents. 
\end{proof}

\begin{table}

\begin{enumerate}
  \item The development of the K-ILO model involved the assumption $\Lambda_{1,0} = 0$ 
  (\emph{i.e.} zero DC component of $\lambda_1$, see \cref{sec1a}). 
  For circuits with $\Lambda_{1,0} \neq 0$  
  the K-ILO spectrum might differ significantly from the correct FP ILO-PMM solution. 
  The exact issues causing, $\Lambda_{1,0} \neq 0$, for a given circuit 
  are not always entirely understood or easily described. This 
  is especially true for higher dimensional ($n>2$) systems.
  It is in many cases a consequence of strongly non-linear circuit nodes 
  which induce DC components in the corresponding dual Floquet vector, 
  $v_1(t)$, (see text accompanying \cref{sec2a1:eq1,sec2a1:eq2}). Furthermore, 
  a non-constant noise modulation matrix, $B(x_s(t))$, often 
  induces the aforementioned result due mixing products landing at DC. 
  \item The K-ILO model further assumes that the DC and higher-order terms 
  of $\lambda_2$ be zero \emph{i.e.} $\Lambda_{2,i} = 0$, for $i\neq \pm 1$. 
  Loosely speaking, this scenario could be induced by simply by the presence of 
  non-sinusoidal nodes in the circuit PSS vector. 
  In other words, non-sinusoidal PSS nodes could, to a certain extent at-least, 
  induce non-sinusoidal nodes in $\lambda_2$. This relation, described here, 
  is only based on empirical observation and is in no way 
  to be considered a scientific law. It is 
  not generally not known how the harmonic content of a linear-response 
  operator, like $\lambda_2$, is connected to the underlying PSS, especially 
  not for higher dimensional circuits. Nevertheless, in  \cref{sec3} below, 
  an example is discussed where a noise-source, injected into a 
  common-ground node (non-sinusoidal PSS node) of a cross-coupled CMOS 
  circuit, results in a non-sinusoidal $\lambda_2$ node.  
  \item The original ILO-PMM model describes the spectrum as 
  an infinite sum of Lorentzian spectra (see \crefrange{sec1b:eq1}{sec1b:eq5}) whereas the standard characteristic 
  only involved the sum of two such spectra (see \cref{sec2:eq1,sec2a1:eq1,sec2a1:eq2}). 
  As was discussed in \cref{app1}, for a strong noise drive we have 
  $|\mu_2| \sim \omega_0^2 c$ (same order of magnitude) and the poles of the 
  Lorentzians in the sum separate and each term in \cref{sec1b:eq1} must 
  be included separately. 
  Here, the reduced-order K-ILO model will, once again, diverge from the 
  correct response. For 
  weak noise drive, which is the case consider herein, $|\mu_2| \gg \omega_0^2 c$ 
  and the poles are more-or-less constant for all indices of interest (for larger 
  indices the terms in \cref{sec1b:eq1} will approach zero with the 
  nominator of the fractions) and we regain the single-pole characteristic 
  predicted in \cref{sec2a1:eq1}.      
\end{enumerate}
\caption{Comparing the ILO-PMM, \crefrange{sec1b:eq1}{sec1b:eq5}, 
to the reduced-order K-ILO model, \cref{sec2a1:eq1,sec2a1:eq2}. Note that 
since the K-ILO is an equivalent representation of the original 
Kurokawa model\cite{kurokawa1968} all conclusions apply equally to this methodology (see 
\cref{sec2b:lem1}). }
\label{sec2b:tab1}
\end{table}

\section{Numerical Experiments}
\label{sec3}

\begin{figure}[!h]
\begin{center}
\includegraphics[scale=0.85]{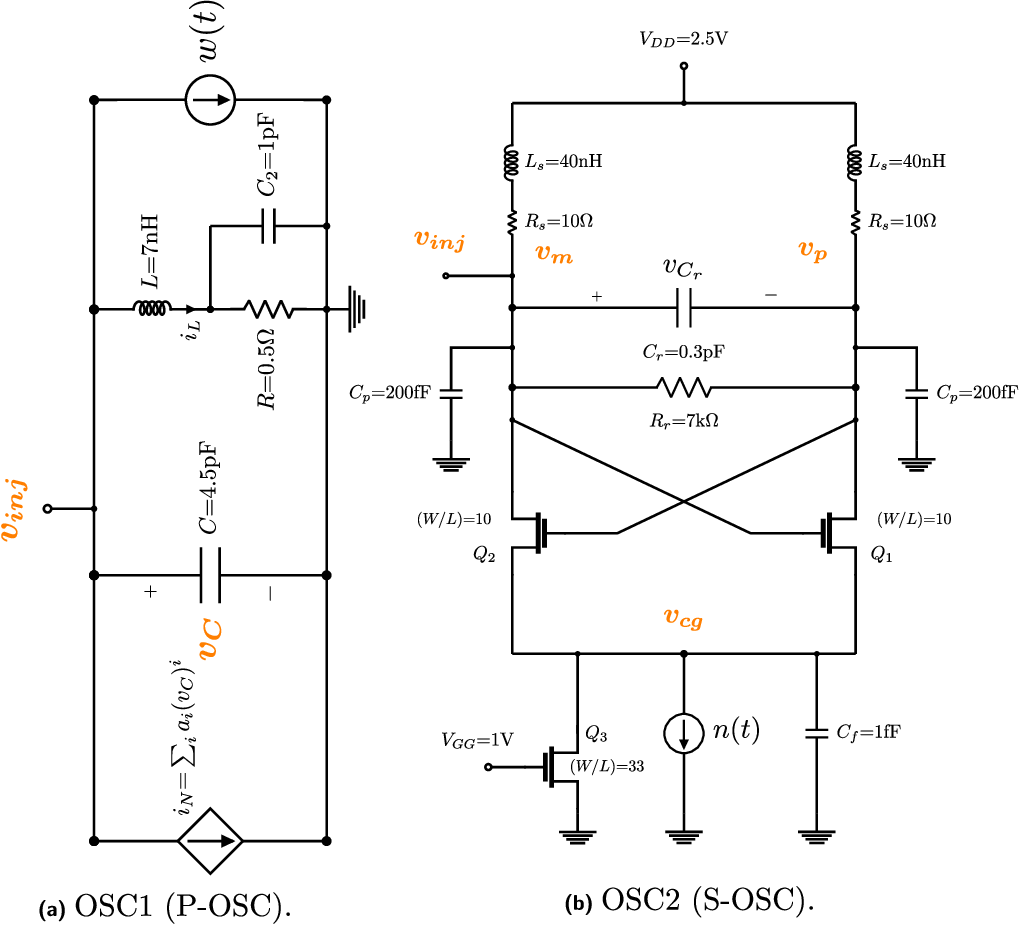}
\end{center}
\caption{ The figure shows the  two oscillator units which 
are coupled, as shown in \cref{sec1a:fig1}.(b), to produce the 
ILO circuit discussed here. The unilateral buffer amplifier (see \cref{sec1a:fig1}.(b)), 
which connects to two circuits through the $v_{inj}$ ports, has a 4th order polynomial in-out characteristic $v_{out} = \sum_{k=0}^3 g_{ck}(v_{in})^k$.
Component values are fixed to values shown unless otherwise stated. 
(a) : OSC1 (P-OSC) is a simple LC negative-resistance oscillator.
Free-running oscillation frequency $f_0 = 900.9 \mathrm{MHz}$, corresponding 
to a period of $T_0 = 1.11\mathrm{ns}$. VCCS (negative resistor) parameters :  
$a_1 = -1.0\mathrm{mS}$, $a_3 = 100\mathrm{\mu A/V^3}$ and $a_i = 0$ for $i \neq 1,3$. 
The white-noise current source, $w(t)$, has the rms current strength 
$i_{\text{w,rms}} = 1\mathrm{pA/\sqrt{Hz}}$.
(b) : OSC2 (S-OSC) is a CMOS cross-coupled LC-tank circuit. 
Free-running oscillation frequency $f_0 = 892.86 \mathrm{MHz}$, corresponding 
to a period of $T_0 = 1.12\mathrm{ns}$. The CMOS devices are considered noiseless.
The rms current of white-noise source, 
$n(t)$, is given as $i_{\text{n,rms}} = 70.7\mathrm{pA/\sqrt{Hz}}$. 
CMOS model (all transistors) : 
$V_{\text{th0}}=0.5\mathrm{V}$, $\lambda = 0.05\mathrm{V^{-1}}$, $k_p = 120\mu \mathrm{A/V^2}$ 
which is also the model used in \cite{Maffezzoni13}. The CMOS channel width-to-length ratios 
are listed next to the respective components. } \label{sec3:fig1}
\end{figure}

Consider the ILO circuit created by coupling the two oscillator units, 
shown in \cref{sec3:fig1}, according to the diagram in 
\cref{sec1a:fig1}.(b). Here the LC oscillator, OSC1, 
is connected unilaterally, through a buffer amplifier, to the cross-coupled 
CMOS unit, OSC2. Following the schematic in \cref{sec1a:fig1}.(b), 
OSC1, the P-OSC, is free-running while OSC2, the S-OSC, assuming a synchronized PSS is reached, will be locked 
to the injected signal. The circuit described here oscillates at a frequency around 
$f_0 = 1/T_0 = 0.9 \mathrm{GHz}$. Herein, this circuit
will be used in a series of simulation trials aimed at investigating and verifying 
the novel ILO-PMM framework, which was developed above in \cref{sec1b}.  
A study is conducted, comparing the spectrum, calculated using the 
ILO-PMM framework, with the  
results produced by the \emph{pnmx} numerical routine; a module contained in the 
commercial Keysight-ADS\textsuperscript{\textcopyright} EDA suite.
Specific attention is furthermore paid to the points raised 
in \cref{sec2b:tab1,sec2b:lem1}. The purpose of this exercise 
is to explore the operational bounds of the phenomenological Kurokawa 
methodology, represented herein by the reduced-order K-ILO equivalent model developed in 
\cref{sec2a1:lem1}. 
 \par
\Cref{sec3:fig2}.(a) shows the ILO PSS for a select number of 
circuit nodes of oscillator units shown in \cref{sec3:fig1}.
Here, the oscillator units are coupled linearly with $g_{c1} = 35.0 \mu \text{A/V}$ 
(see \cref{sec3:fig1} caption). From this  figure it follows that 
node $v_{cg}$ contains significant higher harmonic content (\emph{i.e.} 
2nd, 4th \emph{etc.}) as this is a common-ground node. The 
PSS at this node does therefore \underline{not} conform to the 
near-sinusoidal PSS description on which the K-ILO reduced-order 
model is built (see discussion in \cref{app1}). In \cref{sec3:fig2}.(b), the 
ILO-PMM PNOISE spectral-density, 
derived from \cref{sec1b:eq1}, is plotted for two parameters sets (see figure caption), 
together with corresponding curves produced by the \emph{pnmx} numerical routine. 
Inspecting \cref{sec3:fig2}, the solutions of these two, 
qualitatively very different, numerical algorithms (ILO-PMM/\emph{pnmx}) 
seem to overlap for the most of the frequency range, further verifying 
the novel FP model proposed herein. 

\subsection{Exploring Point \#1 of \Cref{sec2b:tab1}.}

The circuit in \cref{sec3:fig1} does not involve modulated noise 
meaning that the noise-modulation matrix $B \in \mathbb{R}^{n\times p}$ 
is constant. The $p$ columns of, $B$, each correspond to one of the resistors or external 
current noise sources in this circuit. Let 
$j_{w,n} \in [1,p]$ be the noise-indices corresponding to the 
two external noise sources $w(t),n(t) : \mathbb{R} \to \mathbb{R}$ seen in \cref{sec3:fig1}. These two noise  
These two sources dominate the noise response, described in-terms of the lambda functions 
$\lambda_i(t) = v_i^{\top}(t)B : \mathbb{R} \to \mathbb{R}^p$, 
$i=1,2$ (see discussion in \cref{sec1a,sec1b}). 
Let $\lambda_{i,j}(t)$ denote the $j$th component of this 
$p$-dimensional vector, then \cref{sec3:fig3}.(a) plots 
$\lambda_{1,j_w}(t)$, $\lambda_{2,j_w}(t)$, $\lambda_{2,j_n}(t)$. 
Note that the S-OSC (OSC2) noise-source, $n(t)$, does not contribute to $\lambda_1$. 
From \cref{sec3:fig3}.(a) it follows that the dominant noise 
node component of $\lambda_1$, \emph{i.e.} $\lambda_{1,j_w}(t)$, 
has zero DC, $\Lambda_{1,0} = 0$. Following the discussion in point \#1 of \cref{sec2b:tab1}, 
this outcome was expected as the OSC1 circuit shown in \cref{sec3:fig1} sinusoidal 
in nature and hence conform to the Q-SINUS (Kurokawa) representation.  
The requirement, $\Lambda_{1,0} = 0$, upon which the reduced-order 
Q-SINUS K-ILO model (see \cref{sec2,sec2b}) was conditioned,
is hence seen to hold for the circuit in \cref{sec3:fig1}. 
At this point it would seem that the ILO circuit considered here complies 
with the  Q-SINUS Kurokawa methodology. This would then suggest that the 
reduced-order K-ILO equivalent model (see \cref{sec2a1:lem1}) could substitute 
for the full FP ILO-PMM description in \crefrange{sec1b:eq1}{sec1b:eq5} 
as a reasonable approximation for th circuit discussed here. 
However, as will be shown below, this assumption is false. 

\subsection{Exploring Point \#2 of \Cref{sec2b:tab1} - the Breakdown of the Kurokawa Approach.}

\Cref{sec3:fig3}.(a) plots $\lambda_{2,j_w}(t)$ and $\lambda_{2,j_n}(t)$, 
representing contribution, due to noise sources 
$w(t),n(t)$, to the vector $\lambda_2$ (see \cref{sec2a}). Here $\lambda_{2,j_w}(t)$ is again is almost a pure sinusoidal. 
However, it is clear that $\lambda_{2,j_n}(t)$ does not fit this 
description. This should come as no surprise 
since $n(t)$ is injected into 
the common-ground node of the cross-coupled CMOS oscillator OSC2 shown 
in \cref{sec3:fig1}.(b). As noted above (see \cref{sec3:fig2}.(a)), 
the PSS at this node does \underline{not} fit 
the near-sinusoidal PSS requirement. Therefore, one 
should also not expect (loosely speaking) the lambda-vector 
components to conform either (see \# 2 of \cref{sec2b:tab1}). 
The observation implies that the reduced-order K-ILO equivalent  
developed in \cref{sec2a1}, should fail to capture the correct response. 
This can be observed in \cref{sec3:fig3}.(b) which plots 
the PNOISE spectral density of the circuit in \cref{sec3:fig1} 
for the linear coupling $g_{c1} = 35.0\mu\text{A/V}$. The figure shows the 
spectrum calculated using the both the ILO-PMM (\crefrange{sec1b:eq1}{sec1b:eq5}) 
and K-ILO (\cref{sec2a1:eq1,sec2a1:eq2}) methods as well as the numerical 
\emph{pnmx} routine. From this figure it is clear that the
K-ILO model fails to capture the spectrum especially for higher offsets which, 
which is due to the contribution of higher-harmonic content of 
$\lambda_{2,j_n}(t)$ shown in \cref{sec3:fig3}.(a). To further prove this point, 
the intensity of the common-ground noise source $n(t)$ is now decreased 
from $70.0 \text{pA/$\sqrt{\text{Hz}}$}$ to $70.7 \text{fA/$\sqrt{\text{Hz}}$}$. 
Naturally, this will diminish the effect of the contribution $\lambda_{2,j_n}(t)$ 
shown in \cref{sec3:fig3}.(a). Instead, 
the contribution due to P-OSC source, $w(t)$, as represented by the component 
function, $\lambda_{2,j_w}(t)$, will dominate the second mode response. From 
\cref{sec3:fig3}.(a), this contribution is highly sinusoidal which 
in-turn implies that modified circuit adheres to the Kurokawa model requirements 
(see \cref{app1}). Consequently, one should expect the reduced-order K-ILO 
model, developed in \cref{sec2a1}, to capture correct response of the modified 
circuit. \Cref{sec3:fig4}, once again, plots the spectrum of the circuit in \cref{sec3:fig1} 
(see \cref{sec3:fig3} caption) but this time with 
$w(t) = 70.7 \text{fA/$\sqrt{\text{Hz}}$}$. As predicted, the simple reduced-order 
K-ILO model now captures the correct PNOISE response. Finally, from 
\cref{sec2b:lem1}, all results and conclusions 
discussed here, related to the reduced-order K-ILO model developed in \cref{sec2a1:lem1}, extend 
directly to the entire class of ILO PNOISE models built using the Q-SINUS 
methodology including the original well-established and highly referenced 
original representation in \cite{kurokawa1968}.

\begin{figure}[!h]
\begin{center}
\includegraphics[scale=1.0]{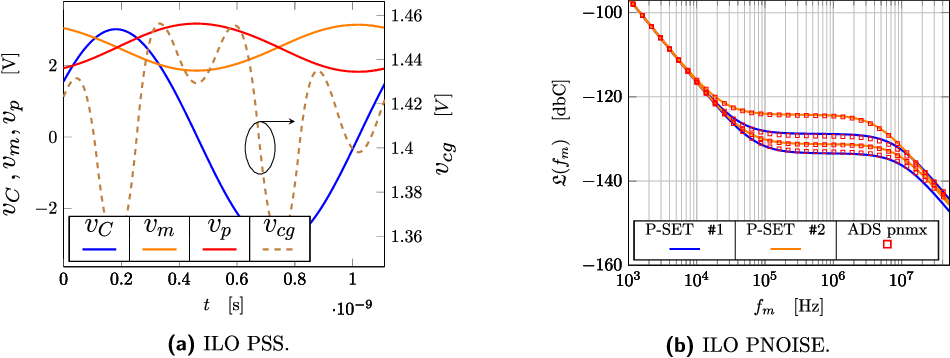}
\end{center}
\caption{ (a) : PSS solution of the ILO circuit shown in \cref{sec3:fig1}.
(b) : phase-noise spectral densities for the ILO circuit in \cref{sec3:fig1}. 
The figure shows the PNOISE spectrum for the ILO circuit 
calculated using the novel ILO-PMM model 
developed herein along with the corresponding output of 
the \emph{pnmx} routine, part of the Keysight-ADS\textsuperscript{\textcopyright} suite. 
The simulations are run for two parameter sets, PSET1 : $(C_r = 0.3035\text{pF},0.295\text{pF}, 
g_{c1} = 35 \mu\text{A/V})$ and PSET2 : $(g_{c1} = 40\mu\text{A/V},60\mu\text{A/V})$ 
with all other component values and circuit parameters fixed as listed 
in \cref{sec3:fig1}.}
\label{sec3:fig2}
\end{figure}

\begin{figure}[!h]
\begin{center}
\includegraphics[scale=1.0]{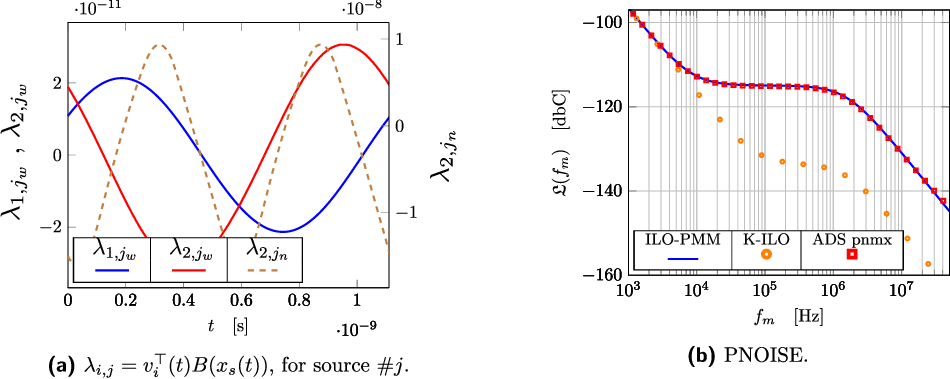}
\end{center}
\caption{ (a) : figure shows the components of $\lambda_i(t) = v_i^{\top}(t)B$, $i=1,2$, 
corresponding to Floquet modes $\mu_1=0,\mu_2 < 0.0$, and representing the  contributions  
due to noise sources $n(t),w(t) : \mathbb{R}\to  \mathbb{R}$ (see \cref{sec3:fig1}) 
which dominate the response. (b) : the PNOISE spectrum calculated using 
ILO-PMM, K-ILO model and the 
\emph{pnmx} routine. Due to the DC and even harmonic components of the contribution 
$\lambda_{2,j_n}(t)$ (see figure (a)) the reduced-order K-ILO model 
fails to predict the  correct spectrum for higher offset frequencies.}
\label{sec3:fig3}
\end{figure}

\begin{figure}[!h]
  \begin{center}
  \includegraphics[scale=1.25]{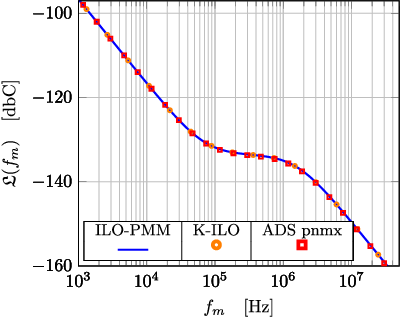}
  \end{center}
  \caption{ PNOISE spectrum of the modified 
  circuit (see discussion in text). The spectrum 
  is calculated using the ILO-PMM, 
  K-ILO model and \emph{pnmx} method. 
  The modified circuit generates a solution which 
  conforms to Q-SINUS/Kurokawa methodology and the reduced-order 
  K-ILO equivalent model is hence able to correctly predict the ILO PNOISE response. 
  This should be contrasted with the result reported in \cref{sec3:fig3}.(b) above.}
  \label{sec3:fig4}
  \end{figure}

\section{Conclusion}

We document the development of a novel time-domain model 
of injection-locked oscillator phase-noise response; the so-called 
ILO-PMM. The methodology is based on a rigorous first-principles approach 
and, as such, is applicable to all circuit topologies, system dimension, 
parameter dependencies \emph{etc.} The companion numerical 
algorithm is compatible with all major commercial circuit simulation programs. 
We explore the analysis of circuits from a macro-model perspective. It is shown that 
the ILO-PMM, developed herein, extends and replaces the standard 
models currently  found in the  literature on this topic. 
The work discussed herein advances the current 
state-of-the-art \emph{w.r.t.} both numerical and theoretical 
modelling and analysis of injection-locked oscillator noise response. 
Several new insights are uncovered with the potential for important practical 
design applications. The analytical innovations and ideas 
discussed herein will serve as inspiration for 
future publications currently being written.

\section*{Acknowledgment}

The authors gratefully acknowledge partial financial support by
German Research Foundation (DFG) (grant no. KR1016/17-1).

\appendix

\section{Proof of \cref{sec2a1:lem1}.}
\label{app1:sec1}

\label{app1}

Below we consider the  concept of a \emph{near-sinusoidal} (N-SINUS) 
vector-function, $y(t) : \mathbb{R} \to \mathbb{R}^n$

\begin{equation} 
y(t) = \alpha  + x(t) + \mathcal{O}(\epsilon)s(t) 
\label{app1:eq0}
\end{equation}

where $\alpha \in \mathbb{R}^n$ is the  DC component of the signal, 
$x(t) : \mathbb{R}\to  \mathbb{R}^{n}$ is a pure sinusoidal vector-function 
(\emph{i.e.} no higher harmonics), $s(t) : \mathbb{R}\to  \mathbb{R}^{n}$ 
is a scaled vector function $|s(t)|\leq 1$ for  all $t$, containing all 
higher harmonics and $|\epsilon| \ll 1$ a small parameter. The 
definition  in \cref{app1:eq0} describes a signal which is \emph{almost} 
sinusoidal which is the case for the PSS in the 
Kurokawa Q-SINUS methodology (see \cite{kurokawa1968,ramirez2008}).   
\par
For an autonomous oscillating system, such as a locked 
ILO circuit, the special Floquet mode 
$\mu_1 = 0$, is known to exist\cite{demir2000,kartner1990,traversa2011}. 
Furthermore, it is well established \cite{demir2000,kartner1990,traversa2011}, that 
the Floquet vector $u_1(t)$ is a scaled copy of $\dot{x}_s$, \emph{i.e.} the 
time-differential of the PSS, $x_s(t)$. As discussed in the above references, 
the proper scale factor leads to $u_1(t) = \dot{x}_s(t)$.
\par 
The Kurokawa modelling approach \cite{kurokawa1968}, assumes a near-sinusoidal 
(N-SINUS) 
planar ($n=2$) PSS $x_s(t) = \alpha + x_a(t) + \mathcal{O}(\epsilon)s(t)$ (see \cref{app1:eq0} 
for description  of terms). It then follows from the  above discussion that 
$u_1(t) = \dot{x}_s = x_b(t) + \mathcal{O}(\epsilon)g(t)$ where 
$x_b(t) = \dot{x}_a(t)$ is sinusoidal and $g(t)= \dot{s}(t)$. 
For the ILO configuration we are considering two modes $\mu_1 = 0, \mu_2 < 0.0$ (see 
\cref{sec1b}). Due to 
the Floquet bi-orthogonality condition ($v_i^{\top}(s)u_j(t) = \delta_{i,j}\delta(t-s)$) 
\cite{demir2000,djurhuus2009}, 
we get that $v_1^{\top}(t)x_b(t) = 1$ and $v_2^{\top}(t)x_b(t) = 0$, for all $t$, 
to within an error of order $|\epsilon|\ll 1$. However, since $x_b(t)$ is a sinusoidal 
planar (2D) solution this must imply that $v_i(t)$, $i=1,2$, are both N-SINUS 
vector-functions. The standard formulation of the Kurokawa model 
\cite{kurokawa1968,ramirez2008} does not involve modulated sources meaning that 
the noise matrix is constant, $B \in \mathbb{R}^{2\times p}$, and 
$\lambda_{i}(t) = v_i(t)^{\top}B$, $i=1,2$, are hence both zero-DC, N-SINUS $p$-dimensional 
vector-functions (see \cref{sec1a:eq4,sec1a:eq5} and accompanying text).
\par
From the above analysis, $\lambda_1(t)$, has zero DC, \emph{i.e.} 
$\Lambda_{1,0} = 0$, which implies that the tensor/matrix expression 
in \cref{sec1b:eq4,sec1b:eq5} reduce to

\begin{align}
\Psi &=  \mathcal{O}(\epsilon) \label{app1:eq1} \\
\Phi_{\rho} &= \frac{\sum_p
  U_{2,p}\Lambda_{2,\rho-p}^{\top}\Lambda_{2,\rho-1}^*U_{2,1}^{\dagger}}{j\omega_0(1-p)
  + 2|\mu_{2}|} + \mathcal{O}(\epsilon) \label{app1:eq2}
\end{align}

The Floquet mode vector, $u_2(t) : \mathbb{R} \to \mathbb{R}^2$, is 
zero-DC, N-SINUS (upto an error of $|\epsilon|\ll 1$) 
which follows from $v_2^{\top}(t)u_2(t) = 1$ and  
$v_2$ being zero-DC, N-SINUS as discussed above. Hence, 
$U_{2,p} = \mathcal{O}(\epsilon)$ for $p \neq \pm 1$. 
Hence, non-negligible terms (terms larger than $\mathcal{O}(\epsilon)$) 
must correspond to indices $p=\pm 1$ in \cref{app1:eq2}. However, 
since we must have $|\mu_2| \ll \omega_0$ (holds for all relevant modes, see 
\cite{Maffezzoni13}) it follows that 
the terms corresponding to $p=1$ will dominate. Furthermore, from the 
above discussion we have $\Lambda_{2,j} = \mathcal{O}(\epsilon)$, for $j \neq -1,1$. 
Combined this leaves us with the expression

\begin{equation}
  \Phi_{\rho} = \frac{K_{-1}\delta_{\rho,0} + K_{1}\delta_{\rho,2} }{2|\mu_{2}|} 
  + \mathcal{O}(\epsilon) 
  \label{app1:eq3}
\end{equation}

where 

\begin{equation}
K_s = U_{2,1}\Lambda_{2,s}^{\top}\Lambda_{2,s}^*U_{2,1}^{\dagger}\label{app1:eq4} 
\end{equation}

Inserting \cref{app1:eq1,app1:eq3,app1:eq4} into \cref{sec1b:eq1,sec1b:eq2,sec1b:eq3}, 
and discarding all negligible terms of order $\epsilon$, the  
following expression for the PNOISE spectrum of the reduced-order K-ILO 
model is produced

\begin{equation}
  \mathfrak{L}_{\text{\tiny ILO }}(\omega_m) \approx
   \frac{ (\omega_0^2c)}{ (0.5\omega_0^2c)^2 + \omega_m^2} +
  \frac{ Z_{-1}}{ |\mu_{2}|^2  + \omega_m^2} + 
  \frac{ Z_1 + (4\omega_0^2c) }{ \bigl(|\mu_2| 
  + \bigl(2\omega_0^2c\bigr) \bigr)^2 +
    \omega_m^2}
    \label{app1:eq5}
\end{equation}

where $Z_s \in \mathbb{C}$ is a possibly complex scalar defined through

\begin{equation}
Z_s  = \bigl[\ K_s \bigr]_{q,q}/\Vert X^{[q]}_{s,1}\Vert^2
\label{app1:eq5x}
\end{equation}

with $q$ being the observation node (see discussion in \cref{sec1b,sec1a} and 
\cref{sec1b:eq2,sec1b:eq3}). The first term in \cref{app1:eq5} represents the primary oscillator 
(see \cref{sec1a:fig1}.(b)) free-running 
PNOISE spectrum \cite{demir2000,kartner1990} which we write 
$\mathfrak{L}_{P}(\omega_m)$. For most offsets of interest, 
$\omega_m \gg 0.5\omega_0^2c$, as $0.5\omega_0^2c$ is a very small offset typically 
on the order of $1\mathrm{Hz}$ or lower, The following approximation then 
holds \cite{demir2000} $\mathfrak{L}_{m}(\omega_m) \approx  (\omega_0^2c)/\omega_m^2$ and 
using this expression we can carry out the following calculation 

\begin{equation}
  \mathfrak{L}_{P}(\omega_m) (\omega_m^2 + |\mu_2|^2) \approx
  (\omega_0^2c/\omega_m^2) (\omega_m^2 + |\mu_2|^2) = 
  \omega_0^2c +  (\omega_0^2c/\omega_m^2) |\mu_2|^2 = 
  \omega_0^2c +  \mathfrak{L}_{P}(\omega_m)|\mu_2|^2
  \label{app1:eq6}
\end{equation}

Finally, we want to show that $|\mu_2| \gg  
2\omega_0^2c$ which means that the following relation  

\begin{equation}
|\mu_2| \gg 2\omega_0^2c  \Leftrightarrow \frac{|\mu_2|}{\omega_0} \gg 4\pi\frac{c}{T_0}
\Leftrightarrow -\frac{1}{2\pi}\ln(\iota_2) \gg 4\pi\frac{c}{T_0}
\label{app1:eq7}  
\end{equation}

where $T_0$ is the oscillation period ($\omega_0 = 2\pi/T_0$) and 
we have introduced the Floquet characteristic multiplier, 
$\iota_2 \in \mathbb{R}$, through \cite{demir2000,kartner1990,traversa2011,djurhuus2009} 
$\iota_2 = \exp(-|\mu_2|T_0)$. A reasonable estimate for normal 
operation would be $\iota_2  \in (0.95, 0.7)$ since $\iota_2=1.0$ 
corresponds to the uncoupled scenario (see \cite{djurhuus22}). So the minimal value in 
\cref{app1:eq7} is taken for $\iota_2=0.99$ (for this example). Using 
this estimate,  \cref{app1:eq7} gives

\begin{equation}
|\mu_2| \gg 2\omega_0^2c  \Leftrightarrow  \frac{c}{T_0} \ll \ln(0.99)/(8\pi^2) \sim 10^{-4}
\label{app1:eq8} 
\end{equation}

which generally will holds as $c$ is a very small parameter. Of-course, as  
the value of $\mu_2\to 0$ (towards uncoupling) or the value of the phase-diffusion constant, 
$c$, increases (stronger P-OSC noise drive), this relation breaks down very fast. 
Since \cref{app1:eq7} holds, at-least for reasonable values of $\mu_2$ and $c$, 
we can approximate $|\mu_2| + 2\omega_0^2c \approx |\mu_2|$ in the 
denominator of the third term in \cref{app1:eq4}. Using this approximation 
together with the result in \cref{app1:eq6} we can write \cref{app1:eq5} as

\begin{equation}
  \mathfrak{L}_{\text{\tiny ILO }}(\omega_m) \approx 
  \mathfrak{L}_{\text{\tiny K-ILO }}(\omega_m) + \mathcal{O}(\epsilon)
  \label{app1:eq9}
\end{equation}

where

\begin{equation}
  \mathfrak{L}_{\text{\tiny K-ILO }}(\omega_m) = \frac{\mathfrak{L}_{P}(\omega_m)|\mu_2| + \Delta_0^{(K)}}
  { |\mu_{2}|^2  + \omega_m^2} 
    \label{app1:eq10}
\end{equation}

where $\Delta_0^{\text{\tiny (K)}} \in \mathbb{R}$ is the real scalar

\begin{equation}
  \Delta_0^{\text{\tiny (K)}} = Z_{-1} + Z_1 + 5(\omega_0^2c)
  \label{app1:eq11}
\end{equation}

which can also be written (see \cref{app1:eq5x})

\begin{equation}
  \Delta^{\text{\tiny (K)}}_0 =
  2\bigl[\ U_{2,1}\Re\bigl\{\Lambda_{2,1}^{\top}\Lambda_{2,1}^*\bigr\}
  U_{2,1}^{\dagger} \bigr]_{q,q}/\Vert X^{[q]}_{s,1}\Vert^2 + 5\omega_0^2c
  \label{app1:eq12}
\end{equation}

\end{document}